\documentclass[12pt]{article}
\usepackage{grffile}
\usepackage[colorlinks=true, citecolor=blue]{hyperref}
\usepackage{enumerate,geometry}
\usepackage{mathtools, amssymb, latexsym, amsmath, amsfonts, amsthm, color}
\usepackage{array, fancyhdr, bm, listings}
\usepackage{algorithm,comment}
\usepackage{algorithmic}
\usepackage{subcaption}
\usepackage{thmtools, thm-restate}
\declaretheorem{theorem}

\theoremstyle{plain}
\newtheorem{lemma}[theorem]{Lemma}

\newtheorem{corollary}[theorem]{Corollary}

\theoremstyle{definition}

\newtheorem{remark}[theorem]{Remark}

\DeclareMathOperator*{\argmin}{arg\,min}

\DeclareMathOperator*{\pr}{\mathbb{P}}
\DeclareMathOperator*{\ept}{\mathbb{E}}

\makeatletter
\newenvironment{breakablealgorithm}
  {
   \begin{center}
     \refstepcounter{algorithm}
     \hrule height.8pt depth0pt \kern2pt
     \renewcommand{\caption}[2][\relax]{
       {\raggedright\textbf{\ALG@name~\thealgorithm} ##2\par}%
       \ifx\relax##1\relax 
         \addcontentsline{loa}{algorithm}{\protect\numberline{\thealgorithm}##2}%
       \else 
         \addcontentsline{loa}{algorithm}{\protect\numberline{\thealgorithm}##1}%
       \fi
       \kern2pt\hrule\kern2pt
     }
  }{
     \kern2pt\hrule\relax
   \end{center}
  }
\makeatother

\providecommand{\keywords}[1]{\textit{Keywords:} #1}
\title{Distributed Matrix Tiling Using A Hypergraph Labeling Formulation\footnote{This research was supported in part by the  DTIC contract FA8075-14-D-0002/0007}}

\date{}

\author{Avah Banerjee, Guoli Ding,Maxwell Reeser}
\newcommand{\FormatAuthor}[3]{
\begin{tabular}{c}
#1 \\ {\small\texttt{#2}} \\ {\small #3}
\end{tabular}
}
\author{
\begin{tabular}[h!]{lcr}
   \FormatAuthor{Avah Banerjee\footnote{Formerly Indranil Banerjee}}{banerjeeav@mst.edu}{Missouri S\&T\footnote{Part of this research was done while the author was visiting LSU as a postdoc}}
&   \FormatAuthor{Guoli Ding}{ding@math.lsu.edu}{Louisiana State University}
&\FormatAuthor{Maxwell Reeser}{maxwellr96@gmail.com}{Louisiana State University}
\end{tabular}
}

\begin{document}

\maketitle

\begin{abstract}
Partitioning large matrices is an important problem in distributed linear algebra computing (used in ML among others). 
Briefly, our goal is to perform a sequence of matrix algebra operations in a distributed manner (whenever possible) on these large matrices.
However, not all partitioning schemes work well with different matrix algebra operations and their implementations (algorithms). This is a type of data tiling problem.
In this work we consider a theoretical model for a version of the matrix  tiling problem in the setting of hypergraph labeling.
We prove some hardness results and give a  theoretical characterization of its complexity on random instances.
Additionally we develop a greedy algorithm and experimentally  show  its efficacy.  
\end{abstract}


\tableofcontents

\maketitle
\section{Introduction}
Our problem is motivated by the following. Machine Learning and Scientific Computing usually involve linear algebra operations over large matrices and tensors (elements) (\cite{roberts2019tensornetwork,langley1996elements}). 
To achieve scalability, operations involving these elements are usually carried out using distributed algorithms.
If the involved elements are too large to be stored within a single shared memory system, then distribution is the only viable option in most cases.
In this setting, the problem of partitioning data elements across a collection of nodes over which the computation will be carried out emerges as a problem whose solution can yield significant benefits.

First, we give an informal description of the tiling problem.
We consider a user program $\cal P$ as a high-level collection of operations involving large elements. We consider only matrices and vectors; however, our formulation can be extended to higher dimensions without great difficulty.
These operations may be logically dependent, which is given by a dependency graph $\cal G$.
We want to execute the operations (in $\cal P$) in a distributed manner on a set of computational nodes. 
In general, for different operations, we may have one or more distributed algorithms implementing the operation.
For example, suppose we have several different distributed implementations of matrix multiplication, which takes two input matrices and returns their product.
This operation can be implemented using multiple distributed algorithms (e.g., Cannon's Algorithm, Distributed Stressen's \cite{ballard2012communication}, PUMMA \cite{choi1994pumma} etc.)  each may prefer a different type of partitioning scheme for the matrices involved. 
An element may participate in multiple operations, and each operation may introduce a different set of constraints on the preferable partition of the element.
Considering the matrix example again, suppose a matrix $A$ is involved in two different operations: $C = \mathsf{mul}_{cannon}(A,B)$ and $D = \mathsf{inv}_{gj}(A)$.
Further, suppose multiplication has been implemented using Cannon's algorithm, which prefers that the matrices be partitioned block-wise.
On the other hand, a matrix inversion using Gauss-Jordan may prefer the matrix $A$ to be distributed as blocks of columns (column tiling).
Unless we want to keep multiple copies of $A$, the choice of the partitioning scheme will affect the performance of different operations involving $A$. 
This example leads us to a natural optimization problem: given a collection of operations, determine an optimal partitioning scheme for the elements to minimize the communication cost. 

\subsection{Problem Formulation}
In this section, we describe some elements of our model at a high level. In subsequent sections, we adapt it based on the specific result we seek.
We often use the phrase ``user program" to indicate a collection of possibly dependent high-level operations. Abstracting away local operations, external memory read-write, etc. We only concern ourselves with operations in the program involving the distributed matrices.
However, our optimization framework is fairly generic.

\subsubsection{Partitioning Schemes}
First, we discuss the type and the degree of granularity in the partitioning scheme that we consider. 
In general, a collection of matrices (either sparse or dense) can be considered as a  hypergraph where the elements of the matrices are vertices, and an edge indicates if the elements are involved in some operation (here operation refers to atomic operations like sum, comparisons, etc.).
Hence a collection of matrices and dependent expressions gives way to a set of hypergraphs, and the goal is to find an optimal $p$-partition (where $p$ is the number of processing nodes) that minimizes the total number of cut-edges. 
Hypergraph partitioning has been used extensively for partitioning data or the computation (\cite{karypis1999multilevel,ballard2015brief,devine2006parallel}).
This problem is approximation-hard and various heuristic based solvers used in practice are best suited when dealing with one such graph at a time.
Further, determination of the exact communication pattern (and thus the edges) may be non-trivial.

On the other hand, most distributed matrix algebra algorithms use some type of block decomposition (especially for dense matrices). 
Thus it makes sense to look at the partitioning scheme at a higher level, which we call \emph{tilings}.
As an example in figure~\ref{fig1: tiling} three commonly used tilings are shown. A tiling need not be contiguous or necessarily disjoint, and as such, there can be many different tiling types (a parameter of our model discussed later). 

\begin{figure}[ht]
\centerline{\includegraphics[scale=1.4]{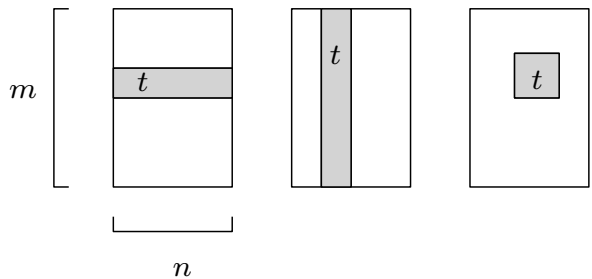}}
\caption{A \emph{tile} is highlighted as the shaded rectangular region. Three types of tilings. From left to right order: row ($r$), column ($c$) and block ($b$).}
\label{fig1: tiling}
\end{figure} 

\subsubsection{Operations}
The second element in our model is the matrix algebra operations.
They are encapsulated at a high level as expressions like $A = \mathsf{mul}(B, C)$. These are the ``atomic expressions" in our model.
So an expression like the following is a composite expression:
\begin{align}\label{eq: exp ex}
   A = \mathsf{sum}(B, C, \mathsf{mul}(D,E,F^T)), 
\end{align}

\begin{figure}[ht]
\centerline{\includegraphics[scale=1.]{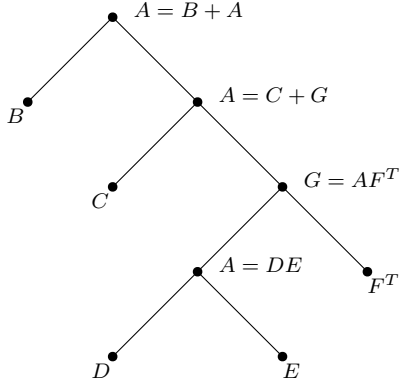}}
\caption{A possible computation DAG corresponding to the expression in Eq~\ref{eq: exp ex}. Here $G$ is an additional matrix to hold the intermediate result $\mathsf{mul}(D, E, F^T)$.}
\label{fig: e-tree}
\end{figure} 
\noindent where $F^T$ is the transpose of $T$.
The above expression does not immediately tell us how we should go about computing the product $DEF^T$.
Interpreting this as $D(EF^T)$ is not the same as $(DE)F^T$ in terms of the number of arithmetic operations used,
since different parenthesization of the matrices in the product term may lead to a differing number of arithmetic operations to compute the final product. This is another optimization issue\footnote{We can solve this easily using dynamic programming formulation.} separate from the partitioning problem.
Also, note that the expression does not explicitly tell us where/how to store this temporary product. 
For example, if the result matrix $DE$ will be used later in some other expression, it may be a good idea to store it as a separate matrix.
This is another orthogonal optimization problem.
To avoid ambiguities in classifying expressions, we consider an abstraction based on hypergraphs (introduced in section 4). However, to prove a lower bound, a simpler model using graphs is considered (section 3).
A possible computation DAG of the expression given in Eq.~\ref{eq: exp ex} is shown in figure~\ref{fig: e-tree}.
Even in the hypergraphic model, we use a partial order to encode the dependency of the expressions similar to using a computational DAG.

\subsubsection{Cost Model}
A solution to our partitioning problem is a tiling of the matrices.
There are various ways to define a cost function based on the communication complexity of the tiling.
A choice of tiling may affect the performance of a distributed algorithm in a non-trivial way.
In most cases, this would require experimental evaluations.
We decouple our cost model from specific system architectures by only considering the abstract cost of the number of \emph{retiling} operations.
Where Retiling is the operation of changing the current tiling of the matrix to meet the algorithm's requirements in the implementation.
\footnote{Alternatively, we may think of this as the cost of accessing non-local memory as if the matrix has been tiled correctly.}
That is, given a tiling of the matrices, we determine the number of instances in which a matrix is not tiled according to the specification of the operation.
If the matrices have unequal sizes, we can use weights to scale the retiling cost accordingly.
In summary, we consider a data partitioning problem on a collection of distributed matrix algebra operations to minimize overall communication, which is approximated by the retiling cost.

\section{Related Work}
The model which is closest to ours\cite{huang2015spartan} introduces a distributed array framework that tries to optimize the tiling (defined at a high level, similar to ours) during runtime. In \cite{zhang2016measuring} the authors develop an array-based distributed framework that tries to optimize the computation DAG to minimize both computation and communication. In \cite{gu2017improving} the authors specifically focused on optimizing matrix multiplication to improve concurrency. Using the Legion programming model\cite{bauer2012legion} the authors in \cite{bauer2019legate} describe a distributed array framework for the popular Numpy Python library. Lastly, the theoretical model we proposed here and our experimental results are currently being adapted to a distributed array processing framework \footnote{ reference redacted in this review copy.}. 

\section{A Signed Graph Model and Approximation Hardness}
In this section, we consider a simpler model to prove the hardness of our tiling problem. In \cite{huang2015spartan} authors gave a similar result showing that their tiling problem is NP-complete (by a reduction from not-all-equal SAT). However, we use a different reduction which is approximation-preserving. This helps us establish an approximation hardness result assuming that the Unique Games Conjecture (\cite{khot2002power,khot2005unique}) is true.

Here we assume that the user program is given as a  directed acyclic graph (DAG).
A program ${\cal P}(V, E)$ is given by an ordered sequence of expressions $E = (e_1,\ldots,e_m)$ along with a set of matrices $V$ ($|V|=n$) \footnote{Later in section 4 we will treat the expressions as edges of a hypergraph.}\footnote{ In what to follow, we will use the terms ``expression" and ``edge" interchangeably. Similarly, we will use the terms ``matrix" and ``vertex" interchangeably. }.
Dependencies are inferred from the ordering of the expressions.
Additionally, we are given a subset $O \subset V$  of output matrices.
These are the matrices that stay in memory until the end of the program execution.
Next, we make an important assumption: each matrix appears at the left-hand side (the output) of an expression at most once. 
Consider the tree in figure ~\ref{fig: e-tree} which corresponds to the following sequence of expressions: $$e_1 (A = DE), e_2 (G = AF^T), e_3(A = C+G), e_4(A=A+B).$$
After the execution of the expression $e_3$, $A$ holds the result of $(C + DEF^T)$ and logically this matrix is different from the $A$ used in $e_1$ and $e_2$.
We can make the case that this matrix is different from the previous $A$.
This implies that it may have a different tiling without incurring any additional cost. Thus we could rewrite the above expressions as:
$$e_1 (A = DE), e_2 (G = AF^T), e_3(H = C+G), e_4(I = H+B).$$

\noindent Note that this does not increase the memory requirement since we can always ``forget" any unused matrices that are not in $O$.
Making these restrictive assumptions on the model only makes our hardness result stronger.

\subsection{The Binary Tiling Problem}
Now we are ready to define the problem formally.
We restrict expressions to only allow at most three matrices (e.g.  $A = \mathsf{sum}(B)$ is allowed but $A = \mathsf{sum}(B,C,D)$ not). Cost of an expression is either 1 (if tilings are sub-optimal) or 0 (otherwise).
As an example, let $A = \mathsf{sum}(B,C)$. Say we assume the $\mathsf{sum}$ operation prefers all matrices to have the same tiling (since it is an elementwise operation).
If $B$ and $C$ have different tilings in the solution $S$, say one is row-wise, and the other is column-wise, then a unit of cost is incurred.
Further, we assume there are only two types of tilings (say row-wise and column-wise).
We will refer to this problem as the binary tiling problem, which is formally defined below.
\footnote{The qualifier ``binary" refers to the fact that we only allow two tiling types.}
Two variants are considered to give a separation-type result.
We only allow $A = B$ and $A = B^T$ types of expressions for the first type. This problem is denoted by $\mathsf{BTP}_T$, where $T$ stands for transpose.
For the expression $A = B$, the communication cost is 0 if both matrices have the same tiling. 
On the other hand, for the expression $A = B^T$, the matrices must have differing tilings.
The input size is the number of matrices ($n$) + number of expressions ($m$).
We show $\mathsf{BTP}_T$ has a polynomial-time (in fact linear) algorithm.
For the second type, we also allow the $\mathsf{sum}$ operator (denoted by $\mathsf{BTP}_{T,+}$). 
This simple modification makes the problem approximation hard.

\begin{figure}[h]
\centerline{\includegraphics[scale=1.]{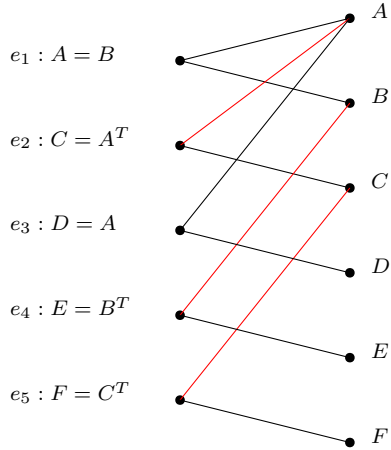}}
\caption{The graph $G$ (right) corresponding to the program given by the expressions (left). Red edges indicate tilings must be different.}
\label{fig: bipart}
\end{figure} 


\begin{theorem}
 $\mathsf{BTP}_T$ can be solved in linear time.
\end{theorem}
\begin{proof}
Consider an input ${\cal P}(V, E)$ to $\mathsf{BTP}_T$.
We can create a bipartite graph $G$ from ${\cal P}(V, E)$ as follows.
The vertex set of $G$ is $V \cup E$. 
There is an edge between $e \in E$ and $A \in V$ if and only if the matrix $A$ is in the expression $e$ (see Fig.~\ref{fig: bipart}).
Note that a matrix never appears more than once on the LHS of an expression and if it is in the LHS of some expression then it must be the first time that matrix appeared in any expression.
Hence $G$ is a tree.
Otherwise, for the sake of contradiction, assume there is some cycle involving the matrices $(A_{i_1},\ldots,A_{i_t})$.
Since there are $t$ expressions there are exactly $2t$ slots, one left and one right for each expression for us to put these matrices. 
Further, each expression must contain two different matrices.
Hence for any ordering of the expressions and assignment of the matrices the matrix appearing in the RHS of the first expression must appear on the LHS of some later expression; due to the pigeonhole principle.
This contradicts our earlier assumption.

\begin{figure}[h]
\centerline{\includegraphics[scale=1.]{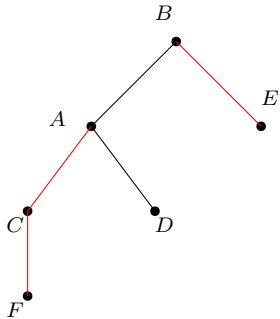}}
\caption{The tree $T$ from the graph $G$ in Fig.~\ref{fig: bipart}}
\label{fig: d tree}
\end{figure} 

Now that we have determined $G$ is a tree it is easy to come up with an algorithm  to optimally tile all the matrices.
Consider the tree $T$ on the vertex set $V$ created from $G$ by adding an edge between a pair of matrices if they were in some expression.
This tree is shown in figure ~\ref{fig: d tree} corresponding to the graph $G$.
We decide a tiling for the root and proceed downward to its children.
Since there are no cycles we never have to backtrack.
Clearly this can be done in linear time.
\end{proof}

As a corollary to the above we see that this restricted tiling problem is solvable in polynomial time as long as $G$ is a tree even with more than two tiling types.
However $G$ must satisfy the condition that a matrix appears at the LHS of an expression exactly once.

\begin{theorem}\label{thm: npc}
 $\mathsf{BTP}_{T,+}$ is NP-Complete.
\end{theorem}

\begin{proof}
Proving the problem is in NP is trivial so we only prove that it is NP-hard by a reducing  from the Balanced Subgraph problem \cite{huffner2007optimal}.

Assume each expression is of the following form : $A = \mathsf{sum}(B^{\pm T}, C^{\pm T})$, where $A^{\pm T}$ denotes either $A$ or $A^T$.
We create a graph $G$ from the expressions as follows.
The vertex set of $G$ is the set of matrices in $V$.
For each expression $A = \mathsf{sum}(B^{\pm T}, C^{\pm T})$ we create two edges.
One between $A$ and $B$ and another between $A$ and $C$.
Additionally, we add signs to these edges.
An edge  has the sign ``$=$" if both $A$ and $B$  are in standard form.
If $B$ is in transpose form then we put the ``$\ne$" sign on the corresponding edge.
The graph $G$ formed this way is 2-degenerate.
In a 2-degenerate graph there is an ordering of the vertices such that every vertex has at most 2 neighbors to its left in the ordering.
We can order the vertices in $G$  as follows.
Let $\mathcal{B}$ be the set of matrices that occur only in the right hand side of an expression.
Let $\mathcal{A}' = V \setminus \mathcal{B}$.
Note that each matrix in $\mathcal{A}'$ corresponds to the expression in which it first occurs (in the LHS).
In our ordering we put the matrices in $\mathcal{B}$ first (in any order) then put the matrices in $\mathcal{A}'$ according to the order of the expressions they first appear.
Since each expression has at most 2 matrices in the RHS it is easy to see that this ordering shows $G$ is 2-degenerate.

$\mathsf{BTP}_{T,+}$ can be restated as a problem of determining a tilling assignment of the vertices in $G$ such that for each $=$-edge the tiling of the incidence vertices match and for each $\ne$-edge the tilings are different.
Then the objective is to find a tiling of the vertices that minimizes the number of unsatisfied edges. Where an edge is said to be unsatisfied if the tilings of its incident vertices do not match with the sign of the edge.
We show this problem to be equivalent to the minimization version of the balanced subgraph problem ($\mathsf{BSP}$) (\cite{huffner2007optimal,dasgupta2007algorithmic}) on a 2-degenerate graph which is defined next.

In $\mathsf{BSP}$ we are given an undirected graph $G(V, E)$ for which we need to find a bi-coloring of the vertices. Associated with each edge is a constraint $=$ or $\ne$. A $=$-edge ($\ne$-edge) is satisfied if it incident vertices have the same (different) color(s).
The goal is to find a coloring that minimizes the number of unsatisfied edges\footnote{In literature some authors uses an alternate but equivalent formulation without using colors on the vertices. Instead a {\it resigning} operation is defined on the vertices which flips the edge types of all the edges incident to the said vertex. Then the goal is to find a sequence of resignings so that the number of edges of the minority type is minimized.}.
A graph is \emph{balanced} if there is a bi-coloring that satisfies all the edges.
The decision problem is to determine for a given $k$ if there is a  bi-coloring such that at-most $k$ edges remain unsatisfied.
This problem is NP-complete(\cite{huffner2007optimal,agarwal2005log}). This is true even for $2$-degenerate graphs which we show next. Any graph $G$ can be transformed to a 2-degenerate graph $G'$ as follows. For each edge in $G$ create a new vertex in $G$ and delete the edge. Then make the new vertex adjacent to the two vertices incident to the deleted edge. 
It is an easy exercise to show that $G'$ is 2-degenerate.
If the deleted edge was a $=$-edge then the two newly created edges are made $=$-edges. Otherwise we make one of the edges a $\ne$-edge arbitrarily. 
We claim that any solution to the  minimum $\mathsf{BSP}$ problem in $G'$ immediately gives a solution to the minimum $\mathsf{BSP}$ for $G$ of the same value $k$.
Let $V'$ the set of new vertices in $G'$ (they replaced the original edges of $G$).
Suppose $C': V\cup V' \to \{0,1\}$ is an optimal bi-coloring on $G'$ that leaves $k$ edges unsatisfied.
Since $C'$ is optimal it cannot leave both edges incident to a vertex in $V'$ unsatisfied. Since we can always flip the color of that vertex to satisfy the two edges incident to it.
This implies that if we restrict $C'$ to $V$ then it induces a coloring on $G$ which also leave $k$ edges unsatisfied. To prove the other direction suppose $C: V \to \{0,1\}$ is a bi-coloring on $G$. We can extend $C$ to create a bi-coloring $C''$ on $G'$ as follows. We let $C''(v) = C(v)$ if $v \in V$ otherwise we let $C''(v) = C(u)$ where $uv$ is a $=$-edge. This ensures that in $G'$ exactly $k$ edges remain unsatisfied. 


To complete the proof we need to reduce  $\mathsf{BSP}_{\mbox{2-degn}}$ to $\mathsf{BTP}_{T,+}$.
Create a matrix for each vertex in $G$. Let $v_1,\ldots, v_n$ be an ordering of the vertices according to the 2-degeneracy structure of $G$.
For each vertex $v_i$ which has a single neighbor $v_j$ ($j < i$) create an expression $A_{i} = A_j$ or $A_i = A_j^T$ depending on whether the sign of the edge is either $+$ or $-$ respectively. Similarly we can deal with case where $v_i$ has two left neighbours. Further it can be easily shown that  $\mathsf{BSP}_{\mbox{2-degn}}$ has a bi-coloring with $k$ unsatisfied edges if and only if $\mathsf{BTP}_{T,+}$ has a tiling with cost $k$.
\end{proof}
 
\begin{corollary}
Assuming the unique games conjecture there are no approximation algorithms for  $\mathsf{BTP}_{T,+}$ with an approximation ratio better than $O(\log n)$.
\end{corollary}

\begin{proof}
The two reductions ($\mathsf{BSP} \le \mathsf{BSP}_{\mbox{2-degn}}\le \mathsf{BTP}_{T,+}$) in Theorem \ref{thm: npc} preserve the size (cost) of the solutions and hence are also approximation preserving.
Then the lower bound follows from the result of
\cite{avidor2007multi} for the minimum $\mathsf{BSP}$ assuming the unique games conjecture \cite{khot2002power}.
\end{proof}

 Another observation of note is that the graph $G$ in the above construction is at most 2-connected (since the rightmost vertex has degree at most 2). This, along with the previous theorem, gives a sharp characterization of our tiling problem with respect to the connectivity of $G$.

\section{Tiling as Hypergraph Labeling}
In this section, we consider a more general formulation of the tiling problem using hypergraphs. This allows us to state an interesting result on the complexity of the problem for random instances. Further, in the following section, we extend this model to develop a greedy algorithm.

Authors in \cite{huang2015spartan} studied the performance of their tiling solver on several randomly generated programs.
However, we suspect that random programs (appropriately defined) may be over-constrained and easier to optimize. Specifically, a random solution may be close to an optimal one. We formally prove this fact in the hypergraph setting introduced in this section. 


Let $H(V, E)$ be a hypergraph whose vertices represent matrices and edges represent expressions. As usual we take $|V| = n, |E| = m$.
We assume $H$ to be $k$-uniform.
Now we define the tiling problem a bit differently.
We do not assume any order on the edges (this does not necessarily make the problem easier). 
There may be one or more algorithms that we can use to execute the expression for each expression.
Each algorithm may have one or more preferred choices of tilings for the matrices involved in the expression. All these preferred choices can be expressed as a constraint on the labeling.
Specifically, we keep a set $L(e)$ for each edge $e \in E$, which is the union of all the preferred tiling configurations of the algorithms that can execute the expression corresponding to the edge.
Suppose we allow at most $\tau$ different tiling types. For example, if we only consider row and column tiling, then $\tau = 2$.
Then each $L(e)$  is a non-empty $\subset [\tau]^k$.
We also use a parameter $s \ge 1$ to denote the number of preferred labelings per edge ($|L(e)| = s$).

Given $(H, L)$ with parameters $k, s$ the optimization problem is to find a labeling $S$ such that:
\begin{align*}
   S \in \argmin_{X \in [\tau]^V} \sum_{e \in E} \left ( \min_{l \in L(e)} d(X,l)  \right)    
\end{align*}

where $d(X,l)$ is defined as follows.
Let $X[v]$ be the label assigned to the vertex $v$. Similarly we define $l[v]$ as the feasible label of the vertex $v \in e$ given by the constraint $l \in L(e)$.
Then
$$d(X,l) = \sum_{v \in e}(1-\delta_{X[v]l[v]})$$
Here $\delta_{ij} = 1$ iff $i=j$ and 0 otherwise.
Hence $d(\cdot,\cdot)$ is  the Hamming distance over the alphabet $[\tau]$. We call this the Constrained Hypergraph Labeling Problem ($\mathsf{CHLP}(H,L)$).

It is an easy observation that the decision version of the problem is $NP$-complete by a reduction from $3\mathsf{SAT}$ with $\tau = 2$. We leave the details as an exercise to the reader. Corollary to this is that verifying whether the optimal cost is 0 is also NP-complete, and hence there is no approximation algorithm with a bounded approximation ratio.

We describe a simple randomized algorithm and show that it achieves a bounded approximation ratio in expectation for a randomly (defined later) generated instance of the problem.
The randomized algorithm, unsurprisingly, is the one that assigns each vertex a label uniformly and independently at random.

\begin{lemma}\label{lmm: upper bound}
Expected cost of the randomized algorithm for any instance of $\mathsf{CHLP}(H,L)$ with parameter $k, s$ is $O(m)$. The result hold with high probability.
\end{lemma}

\begin{proof}
Suppose $S$ is the solution selected at random.
Let,

\begin{align*}
    C(S) = \sum_{e \in E} \left ( \min_{l \in L(e)} d(S,l)  \right)  
\end{align*}
Then the expected cost,

\begin{align*}
   \ept[ C(S)] &= \mathbb{E}[\sum_{e \in E} \left ( \min_{l \in L(e)} d(S,l)  \right)] =  \sum_{e \in E}  \ept[\min_{l \in L(e)} d(S,l)  ]\\
   &\le \sum_{e \in E}  \ept[d(S,l)]
\end{align*}
The last inequality follows from the fact that $\min(x_1,\ldots) \le x_1$ and considering an arbitrary $l \in L(e)$ for each edge $e$. 
Now we can easily compute the expected value using the indicator random variable method. 
For any $v \in e$ let $I_{X[v]\ne l[v]}$ be the event that $v$ is labeled differently between $X$ and $l$.
Then,

\begin{align*}
   \ept[ C(S)]  &\le \sum_{e \in E} \left ( \sum_{v \in e}\pr[I_{X[v]\ne l[v]}]  \right)  = m \left ( \sum_{v \in e}\pr[I_{X[v]\ne l[v]}]  \right) \\
   & = \left(1-\frac{1}{\tau}\right)km = O(m)\ \mbox{when $k$ is fixed.}
\end{align*}
Since $C(S)$ is a sum of $|H|$ i.i.d $0$-$1$ random variables we can apply Chernoff bound to get a high probability result.
Specifically,
\begin{align*}
   \pr[ C(S) \ge (1+\delta)\left(1-\frac{1}{\tau}\right)km]  \le e^{-\frac{\delta^2\left(1-\frac{1}{\tau}\right)km}{3}}
\end{align*}
where $0 < \delta < 1$. This probability tends to 0 as $n \to \infty$ where we assume $m = \Omega(n)$.
\end{proof}

Although the above result in itself is not that interesting, we will need this to give an upper bound on the approximation ratio when used on a random hypergraph. First we need to define a model for random $k$-uniform hypergraphs that are instances of $\mathsf{CHLP}$.

Let $V^{(k)}$ be the set of all $k$-subsets of $V$. A random $k$-uniform hypergraph $H_{n,m,k}$ is then the pair $(V, E)$ where $E \subset V^{(k)}$ of size $m$ chosen uniformly at random from all possible ${{n \choose k} \choose m}$ such subsets.
Then we choose the labeling constraint $L$ as follows.
Assuming each edge $e \in E$ gets exactly $s$ feasible labels, we select a subset $L(e) \subset [\tau]^k$ of size $s$ uniformly  at random from ${[\tau]^k \choose s}$ such subsets.
This gives us a pair $(H_{m,n,k}, L_{\tau,s})$ which behaves uniformly on every labeling $X \in [\tau]^k$ of the vertices. 
Let,

\begin{align*}
   C(H_{m,n,k}, L_{\tau,s}) =  \min_{X \in [\tau]^V} \sum_{e \in E} \left ( \min_{l \in L_{\tau,s}(e)} d(X,l)  \right)    
\end{align*}

be the minimum cost of labeling $H_{m,n,k}$. Due to the minimum at the front it is difficult to determine the expected cost  $\mathbb{E}[C(H_{m,n,k}, L_{\tau,s})]$ over the randomness of the pair $(H_{m,n,k}, L_{\tau,s})$.
However, for the purpose of bounding the approximation ratio of the randomized algorithm we only need to give lower bound of  $C(H_{m,n,k}, L_{\tau,s})$ with high probability.

\begin{lemma}\label{lmm: lower bound}
For some non-negative $t > 0$, $\pr[C(H_{m,n,k}, L_{\tau,s}) > t] = 1-o(1)$ if $m = \Omega(n)$.
\end{lemma}

\begin{proof}
For brevity let $C^* = C(H_{m,n,k}, L_{\tau,s})$. We lower bound the probability $\pr[C^* > t]$.
Let $$Y_X =  \sum_{e \in E} \left ( \min_{l \in L_{\tau,s}(e)} d(X,l)  \right)$$ for each labeling $X \in [\tau]^k$.
Due to the way we have constructed $(H_{m,n,k}, L_{\tau,s})$, $Y_X$'s are i.i.d random variables. Let $Y$ be a r.v. with the same distribution as the $Y_X$'s.
Then,

\begin{align}\label{eq: total cost}
   \pr[C^* > t] =  \left(1- \pr[Y \le t] \right)^{\tau^n} 
\end{align}
Let $Z_e = \min_{l \in L_{\tau,s}(e)} d(X,l)$ for each edge $e \in E$.
Note that $Z_e$'s are i.i.d. and we use the sequence $(Z_1,\ldots,Z_m)$ to enumerate them. 
Note that $Z_i$'s take values between 0 and $k$.
Let $Y = \sum_{i}^m Z_i$ and $Y' = \sum_{i}^m Z'_i$ where 

\begin{align}\label{eq: z}
   Z_i' = \begin{cases}
               0               & \text{if}\ Z_i < k\\
               1               & \text{otherwise}
           \end{cases} 
\end{align}

Then $Z'_i \in \{0,1\}$ for all $i$ and i.i.d.We use $Z'$ to denote an arbitrary $Z'_i$.
According to the above definition $\pr[Y < t] \le \pr[Y' < t/k]$ as the event $[Y < t]$ implies that there are $< t/k$ values of $i$ for which $Z_i = k$. 
$Y'$ is a sum of i.i.d random variables in $\{0,1\}$ and we use Chernoff bound to derive an upper bound on the probability $\pr[Y < t]$ based on the expected value $\mu$ of $Z'$.
For some $0 < \delta' < 1$ we have,

\begin{align*}
    \pr[Y' \le (1-\delta')m\mu] \le e^{-\frac{\delta'^2m\mu}{2}}
\end{align*}
where $\ept[Y'] = m \mu$.
Then,
\begin{align*}
    \pr[Y \le (1-\delta')km\mu] \le e^{-\frac{\delta'^2m\mu}{2}}
\end{align*}
Taking $t = (1-\delta')km\mu$ in Eq.~\ref{eq: total cost}
we get

\begin{align}\label{eq: C}
   \pr[C^* \ge (1-\delta')km\mu] \ge  \left(1- \tau^ne^{-\frac{\delta'^2m\mu}{2}} \right) = 1- e^{n \ln \tau - \frac{\delta'^2m\mu}{2}}
\end{align}
Now we determine $\mu = \pr[Z'=1]$.
According to our definition in Eq.~\ref{eq: z} if $Z'= 1$ then  $\min_{l \in L_{\tau,s}(e)} d(X,l) = k$. Hence,

\begin{align*}
    \pr[Z'=1] &=  \pr \left[\min_{l \in L_{\tau,s}(e)} d(X,l) = k\right] \\ &= (\pr[d(X,l) = k])^s = \left(1-\frac{1}{\tau}\right)^{ks}
\end{align*}
Let,

\begin{align*}
    f(n) =  n \ln \tau - \frac{\delta'^2m\left(1-\frac{1}{\tau}\right)^{ks}}{2}
\end{align*}
 which is the exponent in the RHS of Eq.~\ref{eq: C}. Since $\tau, \delta', s$ and $k$ are all bounded, for some constant $\beta > 0$ we have $f(n) < -\beta n$ whenever $m = \Omega(n)$. This proves the lemma.

\ifx false

\begin{align*}
    \pr[Z = a] &\\ &= \left(1-{k \choose a}\tau^{-a}(1-\frac{1}{\tau})^{k-a}\right)^s \\&-\left(1-{k \choose a-1}\tau^{-a+1}(1-\frac{1}{\tau})^{k-a+1}\right)^s \\
    & = f(k,s, \tau,a)
\end{align*}
\fi

\end{proof}

Now we use Lemma~\ref{lmm: upper bound} and Lemma~\ref{lmm: lower bound} to prove our main result of this section.

\begin{theorem}
There is a randomized algorithm that with high probability has a bounded approximation ratio, which only depends on $k,s,\tau$, for the class or random hypergraphs $H_{m,n,k}$ with random feasibility constraints $ L_{\tau,s}$.
\end{theorem}

\begin{proof}
The upper tail bound of the randomized algorithm described in Lemma~\ref{lmm: upper bound} applies to any hypergraph, not necessarily random. Hence the high probability results of Lemma~\ref{lmm: upper bound} and Lemma~\ref{lmm: lower bound} are independent. They jointly hold with high probability.
The approximation ratio is $$\le 
\frac{(1+\delta)(1-1/\tau)}{(1-\delta')\mu}$$ which is function of $k,s,\tau$ only for a specific choice of $\delta, \delta'$.
\end{proof}

In the figure~\ref{fig: hist} below we plot the histogram of the cost function $C(\cdot)$ for an $(H,L)$ pair sampled according to our random hypergraph model. The plot supports the theorem; showing the cost is distributed over a somewhat narrow range.  

\begin{figure}[h]
\centerline{\includegraphics[scale=0.3]{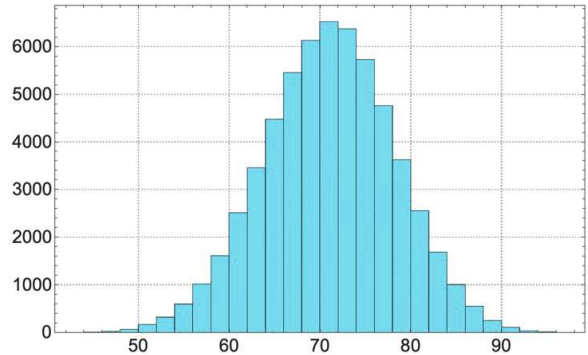}}
\caption{A histogram of $C(\cdot)$ for an instance of $(H_{50,10,3}, L_{3,3})$}
\label{fig: hist}
\end{figure} 
\section{A Greedy Algorithm}
Based on the hypergraph labeling framework introduced in the previous section, we present a greedy algorithm for a more realistic version of the tiling problem.
Essentially, using the greedy heuristic we partition the search space  and iteratively solving the problem  on these subspaces via exhaustive search.
Experimental evaluations are presented in section 6.
 
In this approach, we do not discard the dependency information available in the computation DAG. We use it to develop a greedy order that we use to choose how we process the expressions (edges).
Let $H(V, E)$ be the input hypergraph corresponding to the user program. As before, $V$ constitutes the set of matrices, and $E$ are the edges corresponding to the expressions.
Additionally, we are also given the computation DAG $G$ on the vertex set $V$.
$G$ induces a partial order on $V$. This, in turn, induces a partial order on the edges in $E$.
For two edges $e, e' \in E$, we define $e >_G e'$ if and only if $\exists (A, B) \in e \times e'$ such that 
there is a directed path in $G$ from $A$ to $B$.
Since $G$ is acyclic, if there is a $A\leadsto B$ path, then there can be no $C\leadsto D$ path such that $(C,D) \in e'\times e$.  
Hence the partial order is well defined.
The cost of a tiling operation is defined exactly as before.
For each edge $e$, the cost $C(e, S)$ for the tiling S is the minimum cost over all feasible tiling sets. 
To make our formulation robust, we also allow a weight function $w$ over the edges that enable accounting for things like multiple executions of the expression (inside a loop), unequal sizes of the matrices, and computational complexity of the expression, etc.
Lastly, we do not assume $H$ to be uniform, nor we set any constraint on $s$, the size of the feasibility set.
In summary, the input to our greedy algorithm is the tuple $(H(V,E), G, L, w)$. Recall that $L$ is the feasibility constraints imposed by the algorithms implementing a particular operator. We use $\mathsf{TP}(H(V,E), G, L, w)$ to denote this tiling problem. We will omit some or all of the terms inside the parentheses for notational clarity whenever the meaning is clear.

Now we are ready to describe our algorithm, which has several parts.
First, we decompose $G$ into connected components.
We can solve these components independently of each other.
This is the preprocessing step which is given in algorithm~\ref{alg: hyper}. 

\begin{breakablealgorithm}
\caption{Preprocess $\mathsf{TP}(H, G)$}
\begin{algorithmic}[1]

\STATE Find all the connected components in $G$. Let $\mathcal{C} = \{G_1, \ldots, G_k\}$ be this set. \\
\STATE $ i \leftarrow 0$
\WHILE{$i \le k$}
\STATE $\mathsf{greedy\mbox{-}solver}(\mathsf{TP}(H[G_i], G_i))$\\
\COMMENT{Here $H[G_i]$ is the induced sub-hypergraph with the same vertex set as $G_i$.}
\STATE $i \leftarrow i+1$
\ENDWHILE\\


\end{algorithmic}
\label{alg: hyper}
\end{breakablealgorithm}

Then we process each component independently based on its size. If the size (number of variables + expressions) is ``small" then we compute an optimal tiling by an exhaustive search.
Otherwise, we take a greedy approach.

Henceforth we assume $G$ is connected. 
Let $(L_1,\ldots, L_p)$ be the level sets of the poset $(E, >_G)$ in non-decreasing order of dependency.
That is, expressions in $L_1$ are computed directly from the input matrices and does not depend on other layers, expressions in $L_2$ depend on $L_1$, and so on.
Algorithm~\ref{alg: greedy solver} below describes the outer structure of our algorithm.

\begin{breakablealgorithm}
\caption{ $\mathsf{greedy\mbox{-}solver}(\mathsf{TP}(H, G))$}
\begin{algorithmic}[1]
\STATE {\bf Parameters:} size threshold $\alpha$ 
\STATE Compute the layered decomposition of $E$.
Let ${\cal L} = (L_1,\ldots,L_p)$ be the layers based on this decomposition.\\
\COMMENT{Let ${\cal S}$ is the space of all tilings.}

\IF{$|{\cal S}| \le \alpha$} 
\STATE Perform an exhaustive search on the space  ${\cal S}$
\RETURN {$\mbox{argmin}_{S \in {\cal S}}\sum_{e \in E} w(e)C(e, S) $}

\COMMENT{Otherwise ${\cal S}$ is big and we use a greedy approach.}
\ELSE
\RETURN $\mathsf{inner\mbox{-}greedy\mbox{-}solver}(\mathsf{TP}(H, G),{\cal L})$
\ENDIF

\end{algorithmic}
\label{alg: greedy solver}
\end{breakablealgorithm}
Next we describe the inner solver which uses our greedy method. Initially we start with the full set of edges. The algorithm calls a subroutine that produces a greedy ordering of the remaining set of edges. From this ordering we choose the first $\le \beta$ (a tunable parameter) edges to process. We compute the optimal tiling of the sub-hypergraph induced by these edges and we remove these before the start of the next iteration.
Additionally we maintain a set $\hat{S}$ which stores the vertices that have been already tiled. 
If $\hat{S}$ is  non-empty then the optimal tiling is computed while fixing the tiling of vertices as given in $\hat{S}$.
Here we abuse the notation $\hat{S}$ to indicate both the set of vertices which are tiled as well as the partial tiling.

\begin{breakablealgorithm}
\caption{ $\mathsf{inner\mbox{-}greedy\mbox{-}solver}(\mathsf{TP}(H, G),{\cal L})$}
\begin{algorithmic}[1]
\STATE Let $E' \leftarrow E$\\
\COMMENT{ Initially all vertices in $V$ are un-tiled. Let $\hat{S}$ be a partial tiling of $V$.}
\STATE  Set $\hat{S} \leftarrow \emptyset$
\WHILE{$V\setminus \hat{S} \ne \emptyset$}
\STATE $E''\leftarrow\mathsf{compute\mbox{-}greedy\mbox{-}order}(V, E', \hat{S}, {\cal L})$\\
\COMMENT{${E}''$ are the set of edges in the first bucket based on their cumulative weights.}

\STATE Choose an optimal tiling $S_{E''}$ for the set of vertices in $\bigcup_{e \in E''}e$.
\STATE $\hat{S} \leftarrow \hat{S} \cup S_{E''}$
\STATE $E'\leftarrow E' \setminus E''$
\ENDWHILE\\
\RETURN $\hat{S}$
\end{algorithmic}
\label{alg: solver2}
\end{breakablealgorithm}

Finally, in the following (algorithm~\ref{alg: greedy order}) we describe the procedure to compute the greedy order. 
The algorithm uses a few parameters that can be tuned experimentally. In line 5 we use a new notation $cov(e)$ which is the set of all edges that $e$ is the cover of in the partial order $(E, <_G)$.
Informally, these are the set of expressions which directly depend on the result of the expression $e$. 
Line 11-18 simply choose an appropriate subset of edges based on the greedy order.
A higher $\gamma$-value indicates that the tiling of the vertices in the edge has a bigger influence on the overall solution cost so we should proceed to tile these vertices first. Results in section 6 support this intuition.

\begin{breakablealgorithm}
\caption{ $\mathsf{compute\mbox{-}greedy\mbox{-}order}(V, E', \hat{S}, {\cal L})$}
\begin{algorithmic}[1]
\STATE {\bf Parameters:} $\beta$ for the bucket size, $\eta$ is the weight ratio.
\STATE For each expression $e \in L_p$ compute  $\gamma(e) \leftarrow w(e)\min_{S \in {\cal S}_{|\hat{S}}} C(e, S)$. 

\COMMENT{Next we compute $\gamma(\cdot)$ for all other expression going up layer-wise. Here ${\cal S}_{|\hat{S}}$ is the remaining search space conditioned on $\hat S$.}

\FOR{$i$ from $p-1$ down to $1$}
\FOR{$e \in L_i$}
\STATE $\gamma(e) \leftarrow \min_{S \in w(e){\cal S}_{|\hat{S}}} C(e, S) + \sum_{f \in cov(e)}\gamma(f)$ \\
\ENDFOR
\ENDFOR
\STATE $i \leftarrow 0$
\STATE Sort $E'$ in descending order based on the $\gamma$ values.
\COMMENT{We process each expressions according to this order. We are abusing the notation $E'$ to indicate both a set and an indexed array.}
\STATE $E'' \leftarrow \emptyset$
\WHILE{$|E''| \le \beta$ }
\IF{$\gamma(E'[i]) \ge \eta \gamma(E'[0])$}
\STATE $E'' \leftarrow E'' \cup E'[i]$
\ELSE
\RETURN $E''$
\ENDIF
\STATE $i \leftarrow i+1$
\ENDWHILE
\RETURN $E''$
\end{algorithmic}
\label{alg: greedy order}
\end{breakablealgorithm}

\subsection{Running Time Analysis}
It is easy to see that our greedy algorithm has a polynomial running time in the number of vertices (matrices) $n$.
Here we give a detailed analysis.
In Algorithm ~\ref{alg: hyper} we find the connected components of the computation graph. This takes $O(n+ m)$ times. Note that each expression has a bounded number of matrices, hence number of edges in $G$ is of $O(m)$.
The exhaustive search is performed only if the size of the search space is bounded, hence we can ignore this case in our analysis (line 3-5 in Algorithm ~\ref{alg: greedy solver}).
Now we turn to Algorithm ~\ref{alg: solver2}.
At each iteration of the while loop size of $\hat{S}$ increases by at least 1.
Hence we iterate at most $O(n)$ times.
Choosing an optimal tiling at line 5 costs $O(\tau^\beta) = O(1)$, since $\tau$ and $\beta$ are assumed to be bounded.
Rest of the operations (set union and difference) can be carried out by any off-the-shelf disjoint set data stricture in total $O(n \alpha(n))$ times, where $\alpha(n)$ is the inverse of Ackerman-type function.
This is for all practical purpose we can assume to be linear.
Only things remain is to determine the cost of computing the greedy order in line 4.
So we turn our attention to Algorithm~\ref{alg: greedy order}.
Clearly cost incurred in line 2 is $O(|L_p|)$.
Now let us look at the double-for loops between line 3-7.
The $\gamma(e)$ value is calculated for each edge exactly once. At line 5, computing the sum of cover takes $O(|cov(e)|)$ time.
Hence summing over all $\gamma(e)$ calculations including that in line 2 we get total run-time of all the instructions upto line 7 is $O(\sum_{e} |cov(e)|) = O(m)$ by the argument we made previously. 
Sorting $E'$ in line 7 costs $O(m \log m)$ and the operations on line 11-18 takes constant $(\alpha(n))$ time, since $\beta, \eta$ are bounded, which is dominated by the cost we incur before line 11.
Hence Algorithm ~\ref{alg: solver2} for has a running time of $O(m \log m)$.
Combining this with the previous analysis of Algorithm~\ref{alg: greedy solver} we see that the total runtime of our greedy solver is $O(nm\log m)$.
\section{Experimental Results}
We implemented the greedy algorithm in Python 
\footnote{https://github.com/folshost/TilingSolver} to facilitate its use in Python-based distributed processing API's which have shown significant growth throughout the past few years.
For the experiments, we chose to investigate the performance of our algorithm on a mix of modifications on known algorithms and random algorithms.
We chose these modified algorithms due to the limited number of supported expressions in our implementation of the greedy algorithm.
For known algorithms, we chose an approximate Linear Regression program, an approximate PCA with a 3-round power method for eigenvector determination, a bi-directional power set series of multiplications, and two random programs.
Although our implementation was sensitive to matrix size as a factor in cost calculations, we chose to leave all matrices used in our simulated programs the same size to simplify testing.

We compared our results between the three algorithms: 1) a local solver, 2) exhaustive search, and 3) our greedy algorithm described in the previous section.
According to the execution order specified by the computation DAG, the local solver chooses locally optimal tilings for each expression in a single forward pass of the program. 
The exhaustive search enumerates all possible tilings and finds the lowest cost available. 
To obtain improvements in computation time (i.e., to make exhaustive search tractable), we restricted our expressions to one implementation for each algorithm.

The greedy search algorithm has three configurable parameters.
All of the experimental data was collected using $\alpha=10$.
However, we chose to vary $\beta$ and $\eta$ in a grid search to investigate the effects of these parameters.
Figure \ref{fig11: beta_vs_eta} shows the effects of this variation.
The figure is an average across our five test programs of the max-normalized times yielded by the grid search.
That is, for each program, the grid search yielded a number of times, which were then normalized against the maximum amount of time required for that program.
To aggregate across the five programs, we averaged the five generated grids on an element-wise basis to characterize the effects of $\beta$ and $\eta$ across all programs, giving equal weight to all programs.
The greedy search algorithm gave equal scores for the solutions derived, irrespective of the parameterization of $\beta$ or $\eta$.

From figure \ref{fig11: beta_vs_eta} we can see a significant effect of the $\beta$ (bucket size) parameter, while it would not appear that there is a significant effect due to $\eta$ (weight ratio).

\begin{table} [h!]
\centering
\begin{tabular}{ |c|c|c|c| }
\hline
Test & Local & E Search & Greedy \\ [0.5ex]
\hline
Linear Regression & 0.0012 &    0.3443 & 0.019366   \\
\hline
Parallel PCA      & 0.0020 &  605.7924 & 0.042917   \\
\hline
Power Set         & 0.0042 & 1275.4500 & 0.042180   \\
\hline
Random 1          & 0.0041 &    1.0372 & 0.029209   \\
\hline
Random 2          & 0.0014 &    0.0517 & 0.023032   \\
\hline
\end{tabular}
\caption{Times for search execution (in seconds)}
\label{table:2}
\end{table}

\begin{figure}[ht]
\centerline{\includegraphics[scale=0.5]{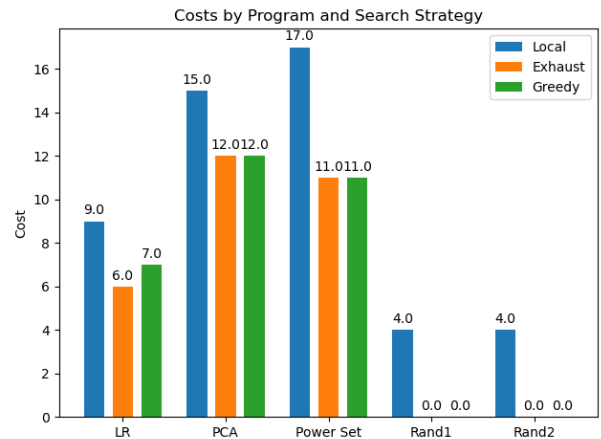}}
\caption{Each program has three searches performed for it}
\label{fig10: tiling}
\end{figure} 

\begin{figure}[ht]
\centerline{\includegraphics[scale=0.5]{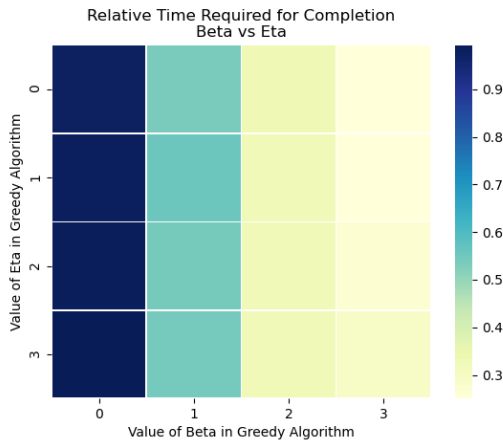}}
\caption{Effects on search time of Beta and Eta parameters}
\label{fig11: beta_vs_eta}
\end{figure} 

In figure \ref{fig10: tiling} we can see that in all cases but one, our algorithm matched the solution found by exhaustive search, and in all cases did better than the local solver. 
Table \ref{table:2} shows us that our algorithm in most cases is comparable in timing to the local solver operation, while the exhaustive search is almost always more than an order of magnitude slower.
For our tests then, our greedy algorithm was generally successful in obtaining the benefits of both of the competing algorithms while retaining none of their drawbacks.

As we noted previously, one of our goals is to make the greedy algorithm work in the distributed HPC environment.
Since this tiling solver would run for every program executed in that environment, it must allow running at such a speed that it does not significantly impact the actual user program's total run time. 
With this in mind, although the exhaustive search would yield good results, in certain cases, the number of involved matrices and operations in a user program could result in a lengthy solve process, as typified by the Parallel PCA and Power Set Multiplication programs. 
All of the programs we used were less than 18 matrix operations. 
For comparison, SOTA deep neural networks like BERT\cite{bert} or GPT-3\cite{brown2020language} often use a large number of layers (in large BERT's case 24 transformer layers, in GPT-3's 96 transformer layers), without taking into consideration any preprocessing steps for data, meaning user programs for investigation could substantially exceed the size of our experimental programs.
We also kept our search space small with the number of implementations of edges in these experiments.
With these facts in mind, in certain circumstances, it could become prohibitive, even with further parallelization of the search process (on server processors or across multiple nodes), to exhaustively search the entirety of that space.

\subsection{Hardware}
For these experiments, we ran all of the tests as multi-threaded processes on an i5-8600k, with a base clock of 3.6 GHz, running in Python 3.6.8 and using Numpy 1.18.1.

\section{A Memory Occupancy Problem}

Here we take a digression and discuss an interesting problem arising out of our tiling optimization study.
Let $G$ be the program DAG and  $H(V,E)$ is the corresponding hypergraph as introduced previously.
Here we consider an optimization of the memory storage by reordering the expressions consistent with the partial order $P$ induced by $G$ on the set $E$ of expression.
{\it Lifespan} of a matrix is the interval starting from the first time it appeared in an expression to the last time. If the matrix is one of the outputs (we denote the set of outputs as $O \subset V$, here  $O$ can be empty if the output of the program is a scalar) of the program then it must be kept in memory at least until the last expression is executed.
The number of matrices that must be kept simultaneously in memory (where the lifespans overlap) depends on the order in which the expressions are executed.
Our goal is to minimize the maximum memory load during the execution of the user program. We will show that this problem is NP hard by a reduction from the cut-width problem. 

We continue to define some more terms.
Let $L$ be a linear extension of $P$. 
For every matrix $ A \not \in O$ let $s_A$ and $t_A$ be the first and the last expression in the ordering $L$ that the matrix $A$ was involved in.
If $A \in O$ then we associate with $A$ the interval $[s_A, m]$ ($m$ is the number of expression).
This forces us to keep the output matrix to stay in memory after it has been computed.
We can create an interval graph $I_{L}$ based on the intervals $[s_A, t_A]$ corresponding to the matrices in $V$ for the linear extension $L$.
If two intervals overlaps then the corresponding matrices must be kept in memory together during the execution of overlapping expressions.
The maximum memory needed to execute the program depends on the maximal set of mutually overlapping intervals for a given ordering of the expressions.
In order to reduce the maximum memory consumption we want to choose a linear extension that minimizes the maximum overlapping set of intervals.
Since output matrices must stay in memory after they have been computed, hence we need to hold at least $|O|$ matrices simultaneously, regardless of the order in which they have been computed.
Let $\kappa(I_L)$ be the clique number of $I_L$. Then the decision version of this memory occupancy problem ($\mathsf{MOP}$) is as follows: given $G, H$ and a positive integer $k$ decide if there is a  linear extension $L$ such that  $\kappa(I_L) \le k$.
\ifx false
As an example consider the following program, where we use the notation $(out, \{input\})$ to indicate the set of input and output matrices involved within an expression\footnote{For this problem the matrix operations can be arbitrary.}.  
{\nonumber \begin{align}
 e_1: (A_{1}; \{A_2, A_3\})\\
 e_2: (A_{4}; \{A_3, A_1\})\\
 e_3: (A_{5}; \{A_1, A_2\})\\
 e_4: (A_{6}; \{A_1, A_2, A_4\})
\end{align}}
where $O = \{A_5,A_6\}$. The intervals for two different ordering of the expressions for the above program is shown in Figure ~\ref{fig: mop}. 

\begin{figure}[h]
\centerline{\includegraphics[scale=1.3]{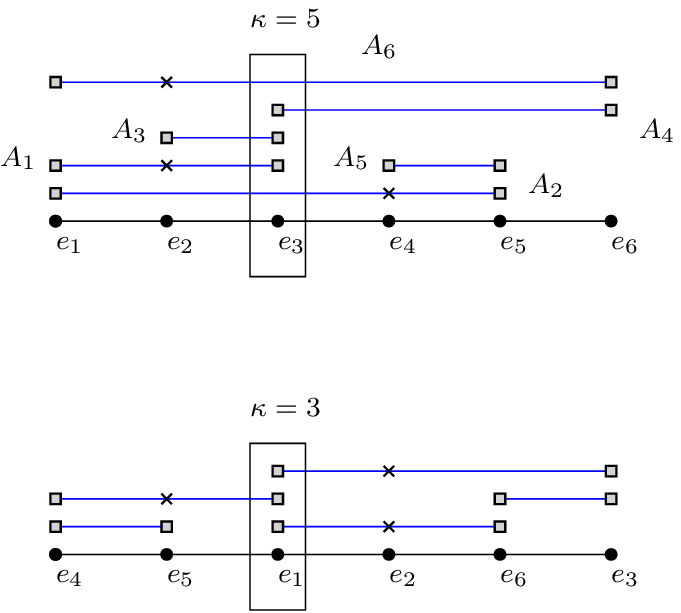}}
\caption{The graph $G_{int}$ corresponding to the program $P_1$. Here the output set is $\{A\}$}
\label{fig: mop}
\end{figure} 
\fi
Next we prove that $\mathsf{MOP}$ is NP-complete.

\begin{theorem}
Deciding whether $\mathsf{MOP}$ has a satisfiable instance for some $k$ is NP-Complete.
\end{theorem}

\begin{proof}

This can be proven by reducing the cut-width problem for simple undirected graphs to our problem.
If vertices of a graph is linearly ordered along a line, then the edges between the points (vertices) on this line forms intervals (see Fig~\ref{fig: cw}).
Let $R$ be some left to right ordering of the vertices. A vertex is denoted by its order from the left. An edge is  represented by an interval $[l,r]$.
We say an interval crosses the $i^{th}$ vertex in $R$ if $i \in [l,r-1]$. Let $\theta_i(R)$ be the number of intervals crossing the $i^{th}$ vertex in $R$.
The cut-width of the ordering $R$ is then $\max_{1\le i\le n} \theta_i(R)$. The \emph{cut-width} of $G$ is minimum cut-width over all possible orderings ($\min_{R \in {\cal S}_n}\max_{1\le i\le n} \theta_i(R)$).
It is known that determining the cut-width of a graph is NP-hard for an arbitrary $k$, but fixed parameter tractable in $k$ (see \cite{thilikos2005cutwidth}).

\begin{figure}[h]
\centerline{\includegraphics[scale=1.]{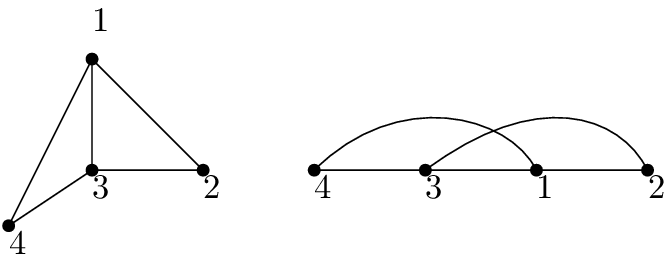}}
\caption{A graph $G$ (left) along with a linear representation (right). Maximum number of edges crossing any vertex is the width, which in this case is 2.}
\label{fig: cw}
\end{figure} 
Given an instance of the cut-width problem (a graph $G(V, E))$ we reduce it to an instance of $\mathsf{MOP}$ in the following way.
Identify a matrix $B_{ij}$ for every edge $(i,j) \in E_v$ and $i, j \in V$ and a matrix $C_i$ for each vertex.
We create an initial expression:
$$e_{-1} : (A_{-1}; \{C_1,\ldots,C_n\})$$
In the above and what follows the matrix before the semicolon is the output matrix and the set of matrices after the semicolon are the input matrices.
For each vertex $i \in V$ create an expression:
$$e_i : (A_i;  \{ B_{ij}\mid j \in N(i)\}\cup \{C_i\}\cup \{A_{-1}\})$$
where $N(i)$ is the set of neighbors of $i$ in $G$.
Now we create an additional expression:
$$e_{n+1} : (A_{n+1} ; \{A_1,\ldots,A_n,A_{-1}\})$$
and set $O = \{A_{n+1}\}$. Let $A = \{A_{-1},A_1,\ldots,A_{n+1}\}$, $B =\{B_{ij}\mid \ (i,j) \in E\} $ and $C = \{C_1,\ldots, C_n\}$.

Let $G'$ be the program DAG determined by the above expressions $\{e_{-1},e_1,\ldots,e_{n+1}\}$ with the vertex set $A\cup B\cup C$.
Next we show that $G'$ has a satisfying instance of size at most $k+n+1$ if and only if $G$ has a cut-width of at most $k$.
Let us prove the only if direction first.
Let $R$ be an ordering of the expressions that produces a satisfying instance of size at most $k+n+1$.
Firstly, In the partial order $P$ induced by $G'$ we have $e_{-1}\succ e_i \succ e_{n+1}$ for all $i \le n$. Further $e_i$ and $e_j$ are incomaparble if $i \ne j$ and $i,j \not \in \{-1,n+1\}$.
Secondly, each expression $e_i$ has a unique (input) matrix $C_i$ associated with it.
The number of matrices from $A \cup C$ that need to be kept in memory just after executing the $i^{th}$ expression  (according to $R$) is exactly $n+1$.
This is independent of the ordering $R$. Hence if $G'$ has a satisfying instance of size at most $n+1+k$ then there are at most $k$ matrices from the set $B$ are kept in memory at any given time. From our construction we see that these are precisely those matrices that corresponds to edges in $G$.
Hence cut-width of $G$ on $R$ is at most $k$.
To prove the other direction assume $G$ has a cut-width of at most $k$ and let $R'$ be an optimal ordering on $V$ ($v_{R'(1)}, \ldots, v_{R'(1)}$). We then extend $R'$ to get an ordering of the expressions $R = (e_{-1},e_{R'(1)}, \ldots, e_{R'(1)},e_{n+1})$ for which we only need to keep at most $k$ matrices from $B$ in memory at any given time. Hence $R$ is an satisfying ordering of size at most $n+1+k$.

\end{proof}

\begin{remark}
In \cite{thilikos2005cutwidth} authors show that the cut-width problem is fixed parameter tractable in $k$. Here we conjecture that $\mathsf{MOP}$ is also fixed parameter tractable if $\min_L \kappa(I_L)$ is bounded. The problem is to determine the cut-width of $G$ when ordering of the vertices are restricted to linear extensions of a given partial order.
\end{remark}
\ifx false
Let $w$ be the width of the partial order $P$. The hardness of $\mathsf{MOP}$ is sensitive to $w$. Consider the two extreme case: 1) when $w = 1$, then $P$ is a total order and the solution is trivial. 2) $w = \Omega(n)$ then $P$ is an anti-chain, and from the proof above we see that $\mathsf{MOP}$ is NP-hard.
This motivates us to come up with an exact algorithm when $w$ is bounded. 
\begin{theorem}
$\mathsf{MOP}$ can be solved exactly in time $O(w^2n)$ when $w$ is bounded.
\end{theorem}
\begin{proof}
We can use dynamic programming here.
Let $R_i$ be an optimal ordering of the expression fr
\end{proof}
\fi




\bibliographystyle{plain}
\bibliography{ref}

\end{document}